\documentclass[smallcondensed,envcountsect,envcountsame]{svjour3} 
\pdfpagesattr{/CropBox [45 160 445 800]}

\usepackage{hyperref,xcolor}

\spnewtheorem{observation}[theorem]{Observation}{\bfseries}{\itshape}
\spnewtheorem{myclaim}[theorem]{Claim}{\itshape}{\upshape}

\title{Linear-Time Recognition of Double-Threshold Graphs\thanks{%
Partially supported by
JSPS KAKENHI Grant Numbers
JP17K00017, 
JP18H04091, 
JP18H05291, 
JP18K11168, 
JP18K11169, 
JP20H05793, 
JP20H05795, 
JP20H05964, 
JP20K11670, 
JP20K11692, 
JP20K20417, 
JP21K11752, 
JP21K11757, 
and
JST CREST Grant Number JPMJCR1402. 
A preliminary version appeared in the proceedings of 
the 46th International Workshop on Graph-Theoretic Concepts in Computer Science (WG 2020),
Lecture Notes in Computer Science 12301 (2020) 286--297.
}}

\author{%
Yusuke Kobayashi \and 
Yoshio Okamoto \and 
Yota~Otachi$^{*}$\thanks{$*$ Corresponding author.} \and 
Yushi Uno
}
\authorrunning{%
Y. Kobayashi,
Y. Okamoto,
Y. Otachi, and
Y. Uno
}

\institute{
Y. Kobayashi \at
Kyoto University, Kyoto, Japan \\
\email{yusuke@kurims.kyoto-u.ac.jp}
\and
Y. Okamoto \at
The University of Electro-Communications, Chofu, Tokyo, Japan \\
\email{okamotoy@uec.ac.jp}
\and
Y. Otachi \at
Nagoya University, Nagoya, Japan \\
\email{otachi@nagoya-u.jp}
\and
Y. Uno \at
Osaka Prefecture University, Sakai, Osaka, Japan \\
\email{uno@cs.osakafu-u.ac.jp}
}

\usepackage{amsmath,amssymb}

\let\doendproof\endproof
\renewcommand\endproof{~\hfill$\qed$\doendproof} 

\usepackage{thmtools}
\usepackage{thm-restate}

\usepackage{algorithm}
\usepackage[noend]{algpseudocode}

\usepackage{graphicx}

\newcommand{\figdir}{./fig}
\newcommand{\figref}[1]{\figurename~\ref{#1}}

\usepackage{color}
\definecolor{darkred}{rgb}{0.65,0,0}

\definecolor{darkblue}{rgb}{0,0,0.9}

\newcommand{\R}{\mathbb{R}}
\newcommand{\lb}{\mathtt{lb}}
\newcommand{\ub}{\mathtt{ub}}
\newcommand{\wei}{\mathtt{w}}
\newcommand{\xcor}{\mathtt{x}}

\date{Submitted on \today}

\begin{document}
\journalname{arXiv}

\maketitle

\begin{abstract}
A graph $G = (V,E)$ is a \emph{double-threshold graph}
if there exist a vertex-weight function $w \colon V \to \R$
and two real numbers $\lb, \ub \in \R$
such that $uv \in E$ if and only if $\lb \le \wei(u) + \wei(v) \le \ub$.
In the literature, those graphs are studied also as 
the pairwise compatibility graphs that have stars as their underlying trees.
We give a new characterization of double-threshold graphs
that relates them to bipartite permutation graphs.
Using the new characterization, we present a linear-time algorithm for recognizing double-threshold graphs.
Prior to our work, the fastest known algorithm by Xiao and Nagamochi~[Algorithmica 2020] ran in $O(n^{3} m)$ time,
where $n$ and $m$ are the numbers of vertices and edges, respectively.

\end{abstract}


\section{Introduction}
\label{sec:introduction}

A graph is a \emph{threshold graph} if there exist a vertex-weight function and a real number called a weight lower bound such that two vertices are adjacent in the graph
if and only if the associated vertex weight sum is at least the weight lower bound.
Threshold graphs and their generalizations are well studied because of their beautiful structures and applications in many areas~\cite{Golumbic04,MahadevP1995}.
In particular, the edge-intersections of two threshold graphs, and their complements 
(i.e., the union of two threshold graphs) have attracted several researchers in the past, and
recognition algorithms with running time 
$O(n^5)$ by Ma~\cite{Ma93},
$O(n^4)$ by Raschle and Simon \cite{RaschleS95}, and
$O(n^3)$ by Sterbini and Raschle~\cite{SterbiniR98}
have been developed,
where $n$ is the number of vertices.

In this paper, we study the class of double-threshold graphs, which is a proper generalization of threshold graphs
and a proper specialization of the graphs that are edge-intersections of two threshold graphs~\cite{JamisonS21}.
A graph is a \emph{double-threshold graph} if there exist a vertex-weight function and two real numbers called
weight lower and upper bounds such that two vertices are adjacent
if and only if the sum of their weights is at least the lower bound and at most the upper bound.
Our main result in this paper is a linear-time recognition algorithm for double-threshold graphs based on a new characterization.

As described below, there are at least two different lines of recent studies that led to this class of graphs:
one is on multithreshold graphs
and
the other is on pairwise compatibility graphs.

\paragraph{Multithreshold graphs.}
Jamison and Sprague~\cite{JamisonS20} introduced \emph{multithreshold graphs}
as a generalization of threshold graphs.
The \emph{threshold number} of a graph $G = (V,E)$ is the minimum positive integer $k$
such that there are $k$ distinct thresholds
$\theta_{1}, \dots, \theta_{k}$ and a weight function $\wei \colon V \to \R$ such that $uv \in E$ if and only if 
the number of thresholds $\theta_{i}$ satisfying $\theta_{i} \le \wei(u)+\wei(v)$ is odd.
Intuitively, the thresholds break the real line into ``yes'' and ``no'' regions
such that two vertices are adjacent if and only if the sum of their weights belongs to a yes region.
Clearly, a graph has threshold number $1$ if and only if it is a threshold graph
and has threshold number at most $2$ if and only if it is a double-threshold graph.
They showed that every graph has threshold number,
and asked some questions including the complexity for recognizing double-threshold graphs.
Puleo~\cite{Puleo20} showed that there is no single choice of three thresholds
that can represent all graphs of threshold number at most $3$.
Jamison and Sprague~\cite{JamisonS21} later focused on double-threshold graphs and showed that
all double-threshold graphs are permutation graphs and
that the bipartite double-threshold graphs are exactly the bipartite permutation graphs.
Our new characterization is closely related to these facts 
and our algorithm uses them.

\paragraph{Pairwise compatibility graphs.}
Motivated by uniform sampling from phylogenetic trees in bioinformatics, Kearney, Munro, and Phillips~\cite{KearneyMP03} defined pairwise compatibility graphs.
A graph $G=(V,E)$ is a \emph{pairwise compatibility graph}
if there exists a quadruple ($T$, $\wei$, $\lb$, $\ub$),
where $T$ is a tree, $\wei \colon E(T) \to \R$, and $\lb, \ub \in \R$,
such that the set of leaves in $T$ coincides with $V$ and 
$uv \in E$ if and only if 
the (weighted) distance $d_{T}(u, v)$ between $u$ and $v$ in $T$
satisfies $\lb \le d_{T}(u,v) \le \ub$. 
Since its introduction, several authors have studied properties of pairwise compatibility graphs, 
but the existence of a polynomial-time recognition algorithm for that graph class has been open.
The survey article by Calamoneri and Sinaimeri~\cite{CalamoneriS16} proposed to look at the class of pairwise compatibility graphs defined on stars
(i.e., star pairwise compatibility graphs), and asked for a characterization of star pairwise compatibility graphs.
As we will see later,
the star pairwise compatibility graphs are precisely the double-threshold graphs (see Observation~\ref{obs:dtg=spc}).

\paragraph{Polynomial-time recognition of double-threshold graphs.}
Xiao and Nagamochi~\cite{XiaoN20}
solved the open problem of Calamoneri and Sinaimeri~\cite{CalamoneriS16} 
by giving a vertex-ordering characterization and 
an $O(n^{3} m)$-time recognition algorithm for star pairwise compatibility graphs,
where $n$ an $m$ are the numbers of vertices and edges, respectively.
Their result also answered the question by Jamison and Sprague~\cite{JamisonS20} about the recognition of double-threshold graphs
by the equivalence of the graph classes.
In this paper, we further improve the running time to $O(m+n)$.

\paragraph{Other generalizations of threshold graphs.}
There are many other generalizations of threshold graphs
such as
bithreshold graphs~\cite{HammerM85},
threshold signed graphs~\cite{BenzakenHW85},
threshold tolerance graphs~\cite{MonmaRT88},
quasi-threshold graphs (also known as trivially perfect graphs)~\cite{YanCC96},
weakly threshold graphs~\cite{Barrus18},
paired threshold graphs~\cite{RavanmehrPBM18},  and
mock threshold graphs~\cite{BehrSZ18}.
We omit the definitions of these graph classes
and only note that 
some small graphs show that
these classes are incomparable to the class of double-threshold graphs
(e.g., 
$3K_{2}$ and the bull for bithreshold graphs,
$3K_{2}$ and the bull for threshold signed graphs,
$2K_{2}$ and the bull for threshold tolerance graphs,
$C_{4}$ and $2 K_{3}$ for quasi-threshold graphs,
$2K_{2}$ and the bull for weakly threshold graphs,
$C_{4}$ and the bull for paired threshold graphs,
$K_{3} \cup C_{4}$ and the bull for mock threshold graphs\footnote{%
The symbols $K_{n}$ and $C_{n}$ denote the complete graph and the cycle of $n$ vertices, respectively.
The disjoint union of two graphs $G$ and $H$ is denoted by $G \cup H$.
For a graph $G$ and a positive integer $k$, 
$kG$ is the disjoint union of $k$ copies of $G$.
The \emph{bull} is a five-vertex path with an additional edge connecting the 2nd and 4th vertices.
It is known that the bull is not a double-threshold graph~\cite{JamisonS21}.}).

Note that the concept of double-threshold \emph{digraphs}~\cite{HamburgerMPSX18}
is concerned with directed acyclic graphs defined from a generalization of semiorders involving two thresholds
and not related to threshold graphs or double-threshold graphs.

\paragraph{Organization of the paper.}
We first review in Section~\ref{sec:preliminaries} some known relationships between double-threshold graphs and permutation graphs,
and then show that connected bipartite permutation graphs admit representations with
some restrictions that we use in subsequent sections.
In Section~\ref{sec:characterization}, which is the main body of this paper,
we give a new characterization of double-threshold graphs.
Using the characterization, we present in Section~\ref{sec:algorithm} a simple linear-time 
algorithm for recognizing double-threshold graphs.

\paragraph{Graph classes.}
In \figurename~\ref{fig:graph-classes},
we summarize the inclusion relations among 
some of the graph classes mentioned so far.
We can see that the class of double-threshold graphs connects
several other graph classes studied before.
\begin{figure}[thb]
  \centering
  \includegraphics[scale=.7]{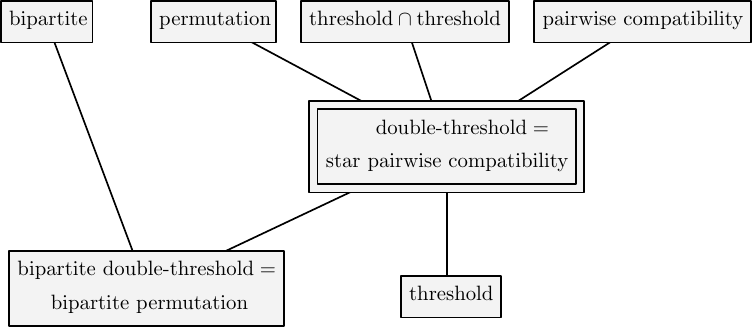}
  \caption{The hierarchy of graph classes mentioned in the introduction. 
      The class of the edge-intersections of two threshold graphs is abbreviated as $\text{threshold} \cap \text{threshold}$.
      A line segment between two classes indicates that 
      the one below is a subclass of the one above.}
  \label{fig:graph-classes}
\end{figure}


\section{Preliminaries}
\label{sec:preliminaries}

All graphs in this paper are undirected, simple, and finite.
A graph $G$ is given by the pair of its vertex set $V$ and its edge set $E$
as $G=(V,E)$.
The vertex set and the edge set of $G$ are often denoted by $V(G)$ and $E(G)$, respectively.
For a vertex $v$ in a graph $G = (V,E)$, its \emph{neighborhood} is the set of vertices that are adjacent to $v$, and denoted by $N_{G}(v) = \{ u \mid uv \in E\}$.
When the graph $G$ is clear from the context, we often omit the subscript.
A linear ordering $\prec$ on a set $S$ with $|S| = n$ can be represented by a sequence $\langle s_{1}, s_{2}, \dots, s_{n} \rangle$ of the elements in $S$,
in which $s_{i} \prec s_{j}$ if and only if $i< j$.
With abuse of notation, we sometimes write ${\prec} = \langle s_{1}, s_{2}, \dots, s_{n} \rangle$.

\subsection{Double-threshold graphs}

A graph $G=(V,E)$ is a \emph{threshold graph}
if there exist a vertex-weight function
$\wei \colon V \to \R$ and a real number $\lb \in \R$ with the following property:
\[
 uv \in E \iff \lb \le \wei(u) + \wei(v).
\]
A graph $G=(V,E)$ is a \emph{double-threshold graph}
if there exist a vertex-weight function
$\wei \colon V \to \R$ and two real numbers $\lb, \ub \in \R$ with the following property:
\[
 uv \in E \iff \lb \le \wei(u) + \wei(v) \le \ub.
\]
Then, we say that the double-threshold graph $G$ is \emph{defined} by $\wei$, $\lb$ and $\ub$.

Jamison and Sprague~\cite{JamisonS20} showed that 
we can use any values as $\lb$ and $\ub$ for defining a double-threshold graph
and that 
we do not have to consider degenerated cases, where 
some vertices have the same weight or
some weight sum equals to the lower or upper bound.
\begin{lemma}
[\cite{JamisonS20}]
\label{lem:nice-representation}
Let $G = (V,E)$ be a double-threshold graph.
For every pair $\lb, \ub \in \R$ with $\lb < \ub$,
there exists $\wei \colon V \to \R$ defining $G$ with $\lb$ and $\ub$
such that $\wei(u) \ne \wei(v)$ if $u \ne v$,
and $\wei(u) + \wei(v) \notin \{\lb, \ub\}$ for all $(u,v) \in V^{2}$.
\end{lemma}

Every threshold graph is a double-threshold graph as one can set a dummy upper bound $\ub > \max \{\wei(u) + \wei(v) \mid u,v \in V\}$.
From the definition of double-threshold graphs,
we can easily see that they coincide with the star pairwise compatibility graphs.
\begin{observation}
\label{obs:dtg=spc}
 A graph is a double-threshold graph if and only if it is a star pairwise compatibility graph.
\end{observation}
\begin{proof}
Let $G = (V,E)$ be a double-threshold graph defined by $\wei\colon V \to \R$ and $\lb, \ub \in \R$.
We construct an edge-weighted star $S$ with the center $c$ and the leaf set $V$
such that the weight $\wei'(vc)$ of each edge $vc \in E(S)$ is $\wei(v)$.
Then, $G$ is the star pairwise compatibility graph defined by $(S, \wei',\lb,\ub)$.

Let $G = (V,E)$ be a star pairwise compatibility graph defined by $(S, \wei, \lb, \ub)$,
where the star $S$ has $c$ as its center.
For each $v \in V$, we set $\wei'(v) = \wei(vc)$.
Then, $G$ is the double-threshold graph defined by $\wei'$, $\lb$, and $\ub$.
\end{proof}

Observation~\ref{obs:dtg=spc} allows us to state the following useful property shown by Xiao and Nagamochi~\cite{XiaoN20} 
in terms of double-threshold graphs.
\begin{lemma}
[\cite{XiaoN20}]
\label{lem:bi-and-nonbi-components}
A graph is a double-threshold graph if and only if
it contains at most one non-bipartite component and all components are double-threshold graphs.
\end{lemma}

The following simple observation is useful when we conduct a detailed analysis on a specific triple $\wei$, $\lb$, $\ub$
defining a double-threshold graph.
\begin{observation}
\label{obs:sandwich}
Let $G = (V,E)$ be a double-threshold graph defined by $\wei\colon V \to \R$ and $\lb, \ub \in \R$.
If $\wei(x) \le \wei(y) \le \wei(z)$ and $xy, yz \in E$ hold for distinct vertices $x,y,z \in V$, then $xz \in E$.
\end{observation}
\begin{proof}
Since $\lb \le \wei(x) + \wei(y) \le \wei(x) + \wei(z)  \le \wei(y) + \wei(z) \le \ub$, we have $xz \in E$.
\end{proof}

The definition of double-threshold graphs can be understood visually in the plane,
by its so called \emph{slab representation}. 
See \figref{fig:dtg_example1} for an example.
In the $xy$-plane, we consider the slab defined by
$\{(x,y) \mid \lb \leq x+y \leq \ub\}$ that is illustrated in gray.
Then, two vertices $u,v \in V$ are joined by an edge if and only if
the point $(\wei(u),\wei(v))$ lies in the slab.

\begin{figure}[thb]
  \centering
  \includegraphics[scale=.8]{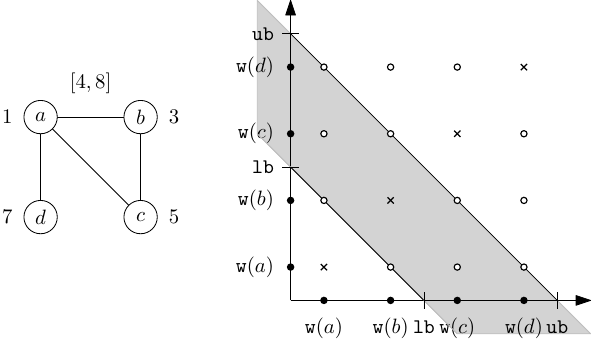}
  \caption{(Left) A double-threshold graph. 
    The weight of each vertex is given as $\wei(a)=1$, $\wei(b)=3$, 
    $\wei(c)=5$, and $\wei(d)=7$; the lower bound is $\lb=4$ and 
    the upper bound is $\ub=8$.
    (Right) The slab representation of the graph.
    A white dot represents the point $(\wei(u),\wei(v))$ for distinct
    vertices $u,v$, and a cross represents the point $(\wei(v),\wei(v))$
    for a vertex $v$.
    Two distinct vertices $u$ and $v$ are joined by an edge if and only
    if the corresponding white dot lies in the gray slab.
  }
  \label{fig:dtg_example1}
\end{figure}

\subsection{Permutation graphs}
A graph $G=(V,E)$ is a \emph{permutation graph}  if there exist linear orderings $\prec_1$ and $\prec_2$ on $V$ with the following property:
\begin{align}
uv \in E \iff (u \prec_1 v \text{ and } v \prec_2 u) \text{ or } (u \prec_2 v \text{ and } v \prec_1 u). \label{eq:perm}
\end{align}
A graph is a \emph{bipartite permutation graph} if it is a bipartite graph and a permutation graph.
It is known that every permutation graph admits a \emph{transitive orientation}~\cite{Golumbic04},
which gives a direction to each edge in such a way that the existence 
of directed edges from $x$ to $y$ and from $y$ to $z$ implies a directed edge from $x$ to $z$ as well.

Jamison and Sprague~\cite{JamisonS21} showed that 
permutation graphs and bipartite permutation graphs have strong connections to double-threshold graphs as follows.
\begin{lemma}
[\cite{JamisonS21}]
\label{lem:DTGisPerm}
Every double-threshold graph is a permutation graph. 
\end{lemma}
\begin{lemma}
[\cite{JamisonS21}]
\label{lem:bipartite}
The bipartite double-threshold graphs are exactly the bipartite permutation graphs.
\end{lemma}

We say that the orderings $\prec_{1}$ and $\prec_{2}$ in \eqref{eq:perm} \emph{define} the permutation graph $G$.
We call $\prec_{1}$ a \emph{permutation ordering} of $G$ if there exists a linear ordering $\prec_{2}$ satisfying the condition above. 
Since $\prec_{1}$ and $\prec_{2}$ play a symmetric role in the definition,
$\prec_{2}$ is also a permutation ordering of $G$.
Note that for a graph $G$ and a permutation ordering $\prec_{1}$ of $G$,
the other ordering $\prec_{2}$ that defines $G$ together with $\prec_{1}$ is uniquely determined.
Also note that if $\prec_{1}$ and $\prec_{2}$ define $G$,
then $\prec_{1}^{\mathrm{R}}$ and $\prec_{2}^{\mathrm{R}}$ also define $G$,
where $\prec_{i}^{\mathrm{R}}$ denotes the reversed ordering of $\prec_{i}$.

We often represent a permutation graph with a \emph{permutation diagram}, which is drawn as follows
(see \figref{fig:permdia_example1} for an illustration).
Imagine two horizontal parallel lines $\ell_1$ and $\ell_2$ on the plane.
Then, we place the vertices in $V$ on $\ell_1$ from left to right according to the permutation ordering $\prec_1$ as distinct points, and similarly
place the vertices in $V$ on $\ell_2$ from left to right according to $\prec_2$ as distinct points.
The positions of $v\in V$ can be represented by $x$-coordinates on $\ell_{1}$ and $\ell_{2}$, which are denoted by $\xcor_{1}(v)$ and $\xcor_{2}(v)$, respectively.  
We connect the two points representing the same vertex with a line segment.
The process results in a diagram (called a permutation diagram) with $|V|$ line segments.
By definition, $uv \in E$ if and only if the line segments representing $u$ and $v$ cross in the permutation diagram,
which is equivalent to the inequality $(\xcor_1(u)-\xcor_1(v)) (\xcor_2(u)-\xcor_2(v)) < 0$.

\begin{figure}[thb]
  \centering
  \includegraphics[scale=.8]{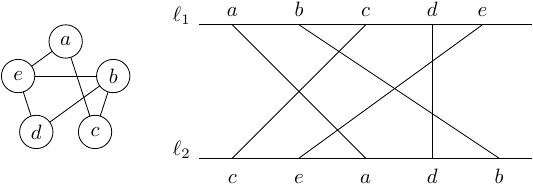}
  \caption{An example of a permutation diagram.
  (Left) A permutation graph $G$.
  (Right) A permutation diagram that represents $G$.}
  \label{fig:permdia_example1}
\end{figure}

Conversely, from a permutation diagram of $G$, we can extract linear orderings $\prec_{1}$ and $\prec_{2}$ as
\begin{align*}
\xcor_1(u) < \xcor_1(v) &\iff u \prec_{1} v,\\
\xcor_2(u) < \xcor_2(v) &\iff u \prec_{2} v.
\end{align*}
When those conditions are satisfied, we say that 
the orderings of the $x$-coordinates on $\ell_1$ and $\ell_2$ are \emph{consistent} with the linear orderings $\prec_{1}$ and $\prec_{2}$, respectively.

Although a permutation graph may have an exponential number of permutation orderings,
it is essentially unique for a connected bipartite permutation graph in the sense of Lemma~\ref{lem:uniquerep} below.
For a graph $G = (V,E)$, linear orderings $\langle v_{1}, \dots, v_{n} \rangle$ and $\langle v'_{1}, \dots, v'_{n} \rangle$ on $V$
are \emph{neighborhood-equivalent} if $N(v_{i}) = N(v'_{i})$ for all $i$.
\begin{lemma}
[\cite{HeggernesHMV15}]
\label{lem:uniquerep}
Let $G$ be a connected bipartite permutation graph defined by $\prec_{1}$ and $\prec_{2}$.
Then, every permutation ordering of $G$ is neighborhood-equivalent to 
$\prec_{1}$, $\prec_{2}$, $\prec_{1}^{\mathrm{R}}$, or $\prec_{2}^{\mathrm{R}}$.
\end{lemma}

A bipartite graph $(X,Y;E)$ is a \emph{unit interval bigraph}
if there is a set of unit intervals $\{I_{v} = [l_{v}, l_{v}+1] \mid v \in X \cup Y\}$ 
such that $xy \in E$ if and only if $I_{x} \cap I_{y} \ne \emptyset$
for $x \in X$ and $y \in Y$.
The class of unit interval bigraphs is known to be equal
to the class of bipartite permutation graphs.
\begin{proposition}
[\cite{HellH04,SenS94,West98}] 
\label{prop:bpg=uib}
A graph is a bipartite permutation graph if and only if 
it is a unit interval bigraph.
\end{proposition}

The following lemma shows that a bipartite permutation graph can be 
represented by a permutation diagram with the special property
that the segments representing vertices of the same set of the bipartition are parallel.
An illustration is given in \figref{fig:lem2-11}.

\begin{figure}[tb]
  \centering
  \includegraphics[scale=.8]{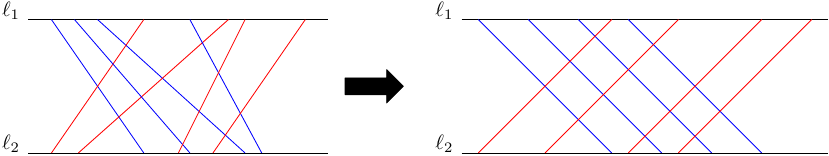}
  \caption{An illustration of Lemma~\ref{lem:rightanglerep}.
  (Left) A permutation diagram of a bipartite permutation graph $G=(X,Y;E)$.
    The vertices in $X$ are represented by blue segments, and
    the vertices in $Y$ are represented by red segments.
  (Right) A permutation diagram that represents $G$ obtained by Lemma~\ref{lem:rightanglerep}.}
  \label{fig:lem2-11}
\end{figure}

\begin{lemma}
\label{lem:rightanglerep}
Let $G=(X, Y; E)$ be a bipartite permutation graph.
Then, $G$ can be represented by a permutation diagram
in which 
$\xcor_{2}(x)=\xcor_{1}(x)+1$ for $x \in X$ and $\xcor_{2}(y)=\xcor_{1}(y)-1$ for $y \in Y$.
\end{lemma}
\begin{proof}
By Proposition~\ref{prop:bpg=uib}, 
there is a set of unit intervals $\{I_{v} = [l_{v}, l_{v}+1] \mid v \in X \cup Y\}$ 
such that for $x \in X$ and $y \in Y$, $xy \in E$ if and only if $I_{x} \cap I_{y} \ne \emptyset$.
We can assume that all endpoints of the intervals are distinct; that is, $l_{u} \notin \{l_{v}, l_{v}+1\}$ for all $u,v \in X \cup Y $ with $u \ne v$~\cite{West98}.
For each $x \in X$, we set $\xcor_{1}(x) = l_{x}$ and $\xcor_{2}(x) = l_{x} + 1$.
For each $y \in Y$, we set $\xcor_{1}(y) = l_{y}+1$ and $\xcor_{2}(y) = l_{y}$.
It suffices to show that this permutation diagram represents $G$.
Observe that line segments corresponding to vertices from the same set, $X$ or $Y$,
are parallel and thus do not cross.
For $x \in X$ and $y \in Y$, we have
\begin{align*}
  I_{x} \cap I_{y} \ne \emptyset
  &\iff
  |l_{x} - l_{y}| < 1 && (\because \text{ all endpoints are distinct})
  \\
  &\iff
  l_{x} < l_{y} + 1 \textrm{ and } l_{y} < l_{x} + 1
  \\
  &\iff
  \xcor_{1}(x) < \xcor_{1}(y) \textrm{ and } \xcor_{2}(y) < \xcor_{2}(x) && (\because x \in X, \  y \in Y)
  \\
  &\iff
  (\xcor_{1}(x)-\xcor_{1}(y)) (\xcor_{2}(x)-\xcor_{2}(y)) < 0.
\end{align*}
The $\Leftarrow$ direction of the last equivalence holds 
since $\xcor_{1}(x) < \xcor_{2}(x)$ and $\xcor_{1}(y) > \xcor_{2}(y)$.
Therefore, we conclude that the diagram represents $G$.
\end{proof}

We can show that for every permutation ordering of a connected bipartite permutation graph,
there exists a permutation diagram consistent with the ordering
that satisfies the conditions in Lemma~\ref{lem:rightanglerep}.
\begin{corollary}
\label{cor:rightanglerep}
Let $G=(X, Y; E)$ be a connected bipartite permutation graph
defined by permutation orderings $\prec_{1}$ and $\prec_{2}$.
If the first vertex in $\prec_{1}$ belongs to $X$,
then $G$ can be represented by a permutation diagram
such that
the orderings of the $x$-coordinates on $\ell_{1}$ and $\ell_{2}$ are consistent with $\prec_{1}$ and $\prec_{2}$, respectively,
and that $\xcor_2(x)=\xcor_1(x)+1$ for every $x\in X$ and $\xcor_2(y)=\xcor_1(y)-1$ for every $y\in Y$. 
\end{corollary}
\begin{proof}
Since $G$ is connected, the last vertex in $\prec_{1}$ belongs to $Y$,
the first vertex in $\prec_{2}$ belongs to $Y$,
and the last vertex in $\prec_{2}$ belongs to $X$.

By Lemma~\ref{lem:rightanglerep},
$G$ can be represented by a permutation diagram $D'$ in which 
$\xcor_{2}(x)=\xcor_{1}(x)+1$ for $x \in X$ and $\xcor_{2}(y)=\xcor_{1}(y)-1$ for $y \in Y$.
Let $\prec'_{1}$ and $\prec'_{2}$ be the permutation orderings corresponding to $\ell_{1}$ and $\ell_{2}$,
respectively, in this diagram $D'$.
Lemma~\ref{lem:uniquerep} and the assumption on the first vertex in $\prec_{1}$ imply that
$\prec_{1}$ is neighborhood-equivalent to $\prec'_{1}$ or $(\prec'_{2})^{R}$.
We may assume that $\prec_{1}$ is neighborhood-equivalent to $\prec'_{1}$
since otherwise we can rotate the diagram $D'$ by $180$ degrees
and get a permutation diagram of $G$ in which the ordering on $\ell_{1}$ is $\prec_{1}$,
$\xcor_{2}(x)=\xcor_{1}(x)+1$ for $x \in X$, and $\xcor_{2}(y)=\xcor_{1}(y)-1$ for $y \in Y$.

Now we can construct a desired permutation diagram of $G$ using $\prec_{1}$ and $D'$
by appropriately giving a mapping between segments and vertices.
That is, for each $i \in \{1, \dots, |X \cup Y|\}$, 
we assign the $i$th vertex in $\prec_{1}$ to the segment in $D'$ with the $i$th smallest $x$-coordinate on $\ell_1$.
This new diagram is a permutation diagram of $G$ since $\prec_{1}$ is neighborhood-equivalent to $\prec'_{1}$.
Since $G$ and $\prec_{1}$ uniquely determine the ordering on $\ell_{2}$,
the $x$-coordinates $\xcor_{2}$ on $\ell_{2}$ are consistent with $\prec_{2}$.
\end{proof}


\section{New characterization}
\label{sec:characterization}

In this section, we present a new characterization of double-threshold graphs (Theorem~\ref{thm:non-bipartite}).
This is one of our main results and a key ingredient of the linear-time algorithm given in the next section.
Recall that Lemma~\ref{lem:bipartite} characterizes the bipartite double-threshold graphs
as the bipartite permutation graphs, which can be recognized in linear time~\cite{SpinradBS87,Sprague95}.
Thus, we are going to focus on \emph{non-bipartite} graphs in this section.

Let $G=(V,E)$ be a graph.
From $G$ and a vertex subset $M \subseteq V$,
we construct an auxiliary bipartite graph $G'_{M}=(V',E')$
defined as follows (see \figref{fig:auxgraph_example1}):
\begin{align*}
    V' &= \{v,\overline{v} \mid v \in V\},
    & E' &= \{u\overline{v} \mid uv \in E\} \cup \{v\overline{v} \mid v \in M\}.
\end{align*}
Note that $(V, \{\overline{v} \mid v \in V\})$ is a bipartition of $G'_{M}$ no matter what $M$ is.

\begin{figure}[tb]
  \centering
  \includegraphics[scale=.75]{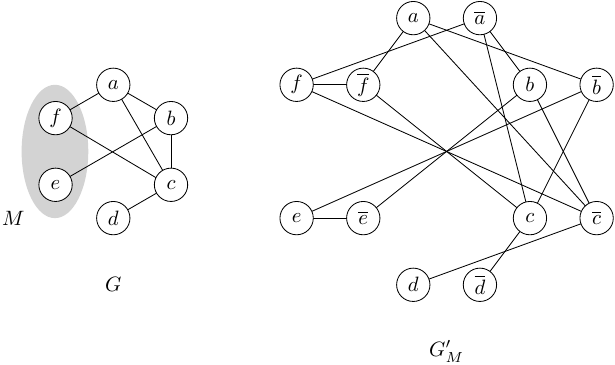}
  \caption{An example of an auxiliary bipartite graph.
    (Left) A graph $G$ and $M=\{e,f\}$.
    (Right) The auxiliary bipartite graph $G'_M$.
    }
  \label{fig:auxgraph_example1}
\end{figure}
An \emph{efficient maximum clique} $K$ of a graph $G$ is
a maximum clique (i.e., a clique of the maximum size)
that minimizes the degree sum $\sum_{v \in K} \deg_{G}(v)$.
See \figref{fig:effcliq_example1}.
\begin{figure}[t]
  \centering
  \includegraphics[scale=.8]{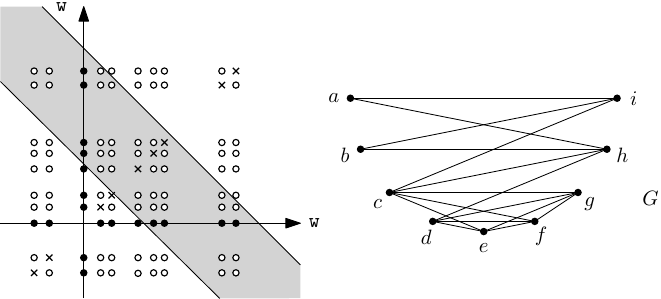}
  \caption{An example of an efficient maximum clique.
    (Left) A slab representation of a double-threshold graph $G$.
    (Right) 
    The vertices of $G$ are ordered in the increasing order of their weights.
    The graph $G$ has two maximum cliques $Q_1=\{c,e,f,g\}$ and $Q_2=\{d,e,f,g\}$.
    The degree sums are $\sum_{v \in Q_1}\deg_G(v) = 5+4+4+4 = 17$, and
    $\sum_{v \in Q_2}\deg_G(v) = 4+4+4+4 = 16$.
    Therefore, $Q_2$ is the only efficient maximum clique of $G$.
    }
  \label{fig:effcliq_example1}
\end{figure}

Using these terms, we present a characterization of \emph{non-bipartite} double-threshold graphs as follows.
\begin{restatable}{theorem}{newcharacterization}
\label{thm:non-bipartite}
For a non-bipartite graph $G$, the following are equivalent.
\begin{enumerate}
  \item $G$ is a double-threshold graph. \label{itm:chara_dtg}
  \item For every efficient maximum clique $M$ of $G$, the graph $G'_{M}$ is a bipartite permutation graph. \label{itm:every_emc}
  \item For some efficient maximum clique $M$ of $G$, the graph $G'_{M}$ is a bipartite permutation graph. \label{itm:some_emc}
\end{enumerate}
\end{restatable}


The rest of this section is devoted to a proof of Theorem~\ref{thm:non-bipartite}.
The following is a quick overview of the proof steps (some terms will be defined later).
\begin{enumerate}
  \item We first prove the key lemma (Lemma~\ref{lem:key-lemma}) 
  ensuring that a graph is a double-threshold graph if and only if 
  $G'_{M}$ is a permutation graph with a ``symmetric'' permutation diagram,
  where $M \subseteq V$ is the set of ``mid-weight'' vertices.
  
  \item We then show that every efficient maximum clique can be the set of mid-weight vertices
  by proving a couple of lemmas (Lemmas~\ref{lem:mid-weight_maximal-clique} and \ref{lem:efficient-max-clique}).
  
  \item Next, we show that the symmetry required in the key lemma follows for free
  if $M$ is a clique (Lemma~\ref{lem:perm=sym}),
  which is true when we set $M$ to be the set of mid-weight vertices.

  \item Finally, we complete the proof of Theorem~\ref{thm:non-bipartite} by putting everything together.
\end{enumerate}

We start with the following simple but useful fact.
\begin{lemma}
\label{lem:GMconn}
For a connected non-bipartite graph $G=(V, E)$ and 
a vertex subset $M \subseteq V$, $G'_{M}$ is connected. 
\end{lemma}
\begin{proof}
For any $u, v \in V$, since $G$ is connected and non-bipartite, $G$ contains both an odd walk and an even walk from $u$ to $v$. This shows that $G'_{M}$ contains walks from $u$ to $v$, from $u$ to $\bar v$, from $\bar u$ to $v$, and from $\bar u$ to $\bar v$. Hence, $G'_{M}$ is connected. 
\end{proof}

For the auxiliary graph $G'_{M} = (V', E')$ of $G=(V,E)$,
a linear ordering on $V'$ represented by $\langle w_1, w_2, \dots, w_{2n} \rangle$
is \emph{symmetric} if $w_i=v$ implies $w_{2n-i+1}=\bar v$ for any $v \in V$ and any $i \in \{1,2,\dots, 2n\}$.
\begin{lemma}
\label{lem:key-lemma}
Let $G=(V,E)$ be a non-bipartite graph and $M \subseteq V$.
The following are equivalent. 
\begin{enumerate}
    \item $G$ is a double-threshold graph defined by $\wei \colon V \to \R$ and $\lb,\ub \in \R$ such that
    $M = \{v \in V \mid \lb/2 \le \wei(v) \le \ub/2\}$. \label{itm:dtg_with_M} 
    
    \item The auxiliary graph $G'_{M} = (V',E')$ can be represented by a permutation diagram
    in which both orderings $\prec_1$ and $\prec_2$ are symmetric.\label{itm:sym_diagram}
\end{enumerate}
\end{lemma}
\begin{proof}
(\ref{itm:dtg_with_M}$\implies$\ref{itm:sym_diagram})
An illustration is given in \figref{fig:lem4-2}.
Let $G$ be a double-threshold graph defined by $\wei\colon V \to \R$ and $\lb, \ub \in \R$
such that $M = \{v \in V \mid \lb/2 \le \wei(v) \le \ub/2\}$.
By Lemma~\ref{lem:nice-representation}, we can assume that $\lb = 0$ and $\ub = 2$,
that $\wei(u) + \wei(v) \notin \{0,2\}$ for every $(u,v) \in V^{2}$,
and that $\wei(u) \ne \wei(v)$ if $u \ne v$.
We construct a permutation diagram of $G'_{M}$ as follows. 
Let $\ell_{1}$ and $\ell_{2}$ be two horizontal parallel lines. For each vertex $w \in V'$,
we set the $x$-coordinates $\xcor_{1}(w)$ and $\xcor_{2}(w)$ on $\ell_{1}$ and $\ell_{2}$ as follows: for any $v\in V$, 
\begin{align*}
    \xcor_1(v)&=\wei(v)-1, & \xcor_1(\bar v)&=1-\wei(v), \\
    \xcor_2(v)&=\wei(v),   & \xcor_2(\bar v)&=-\wei(v).
\end{align*} 
Since $\wei(u) + \wei(v) \notin \{0,2\}$ for every $(u,v) \in V^{2}$ and $\wei(u) \ne \wei(v)$ if $u \ne v$,
the $x$-coordinates are distinct on $\ell_{1}$ and on $\ell_{2}$.
By connecting $\xcor_{1}(w)$ and $\xcor_{2}(w)$ with a line segment for each $w \in V'$, we get a permutation diagram.
The line segments corresponding to the vertices in $V$ have negative slopes,
and the ones corresponding to the vertices in $V' \setminus V$ have positive slopes.
Thus, for any two vertices $u,v \in V$, the line segments corresponding to $u$ and $\bar{v}$ cross
if and only if both $\xcor_1(u) \le \xcor_1(\bar v)$ and $\xcor_2(u) \ge \xcor_2(\bar v)$ hold,
which is equivalent to $0 \le \wei(u)+\wei(v) \le 2$, and thus to $u \bar{v} \in E'$.
Similarly, 
the line segments corresponding to $v$ and $\bar{v}$ cross if and only if $0 \le 2 \wei(v) \le 2$, i.e., $v \in M$. 
This shows that the obtained permutation diagram represents $G'_{M}$.
Let $\prec_{1}$ be the ordering on $V'$ defined by $\xcor_{1}$.
Since $\xcor_{1}(v) = - \xcor_{1}(\bar{v})$ for each $v \in V$, $\prec_{1}$ is symmetric.
Similarly, the ordering $\prec_{2}$ defined by $\xcor_{2}$ is symmetric.

\begin{figure}[tb]
  \centering
  \includegraphics[scale=.8]{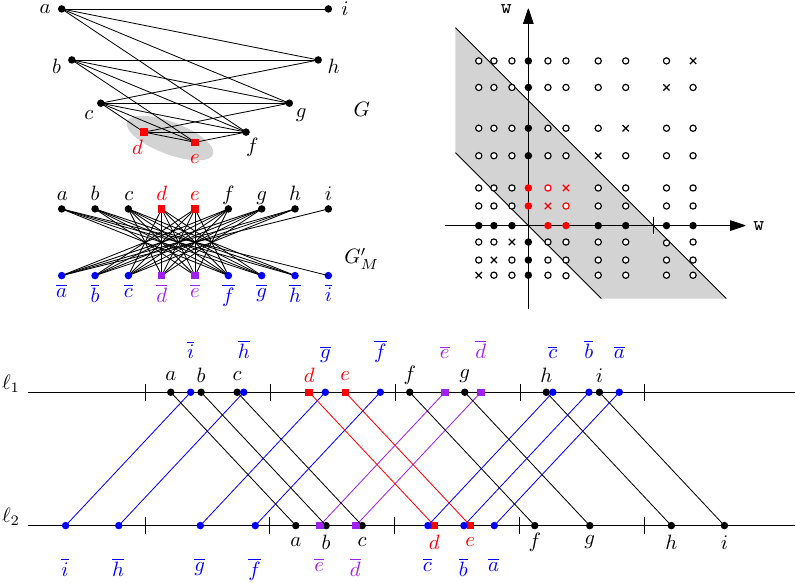}
  \caption{An illustration of (\ref{itm:dtg_with_M} $\implies$ \ref{itm:sym_diagram}) in Lemma \ref{lem:key-lemma}.
  (Top left) A double-threshold graph $G$ with $M=\{d,e\}$. The auxiliary bipartite graph $G'_M$ is also depicted.
  (Top right) A slab representation of $G$.
  (Bottom) A permutation diagram of $G'_M$ as given in the proof.
  }
  \label{fig:lem4-2}
\end{figure}

(\ref{itm:sym_diagram}$\implies$\ref{itm:dtg_with_M})
Suppose we are given a permutation diagram of $G'_{M}$ in which both $\prec_1$ and $\prec_2$ are symmetric.
We may assume by symmetry that the first vertex in $\prec_{1}$ belongs to $V$.
Since $G'_{M}$ is connected by Lemma~\ref{lem:GMconn}, Corollary~\ref{cor:rightanglerep} shows that
we can represent $G'_{M}$ by a permutation diagram in which the $x$-coordinates $\xcor_{1}$ and $\xcor_{2}$ on $\ell_{1}$ and $\ell_{2}$ satisfy that
\begin{equation}
\xcor_2(v)=\xcor_1(v)+1 \quad \mbox{and} \quad \xcor_2(\bar v)=\xcor_1(\bar v)-1  \quad (v\in V) \label{eq:x1x2}
\end{equation}
and that the orderings of the $x$-coordinates on $\ell_{1}$ and $\ell_{2}$ 
are consistent with $\prec_{1}$ and $\prec_{2}$, respectively.
Since $\prec_{1}$ is symmetric, 
if $u, v \in V$ are the $i$th and the $j$th vertices in $\prec_{1}$,
then $\bar{u}, \bar{v}$ are the $(2n-i+1)$st and the $(2n-j+1)$st vertices in $\prec_{1}$.
Since $i < 2n-j+1$ is equivalent to $j < 2n-i+1$,
we have that $u \prec_{1} \bar{v}$ if and only if $v \prec_{1} \bar{u}$.
As $\xcor_{1}$ is consistent with $\prec_{1}$, it holds for $u, v \in V$ that
$\xcor_1(u) \le \xcor_1(\bar v)$ if and only if $\xcor_1(v) \le \xcor_1(\bar u)$,
and hence
\[
  \xcor_1(u) \le \xcor_1(\bar v)
  \iff 
  \xcor_1(u)+\xcor_1(v) \le \xcor_1(\bar v)+\xcor_1(\bar u).
\]
Similarly, we can show that for $u,v \in V$,
\[
  \xcor_2(u) \ge \xcor_2(\bar v)
  \iff
  \xcor_2(u)+\xcor_2(v) \ge \xcor_2(\bar v)+\xcor_2(\bar u).
\]
Thus, for any two distinct vertices $u,v \in V$, it holds that
\begin{align}
  uv \in E
    & \iff u\bar v \in E' \notag \\
    & \iff \xcor_1(u) \le \xcor_1(\bar v) \mbox{ and } \xcor_2(u) \ge \xcor_2(\bar v) \notag \\
    & \iff \xcor_1(u)+\xcor_1(v) \le \xcor_1(\bar v) + \xcor_1(\bar u) \mbox{ and } \xcor_2(u)+\xcor_2(v) \ge \xcor_2(\bar v)+\xcor_2(\bar u). \label{eq:weight01} 
\end{align}
For each $v\in V$, define 
\[
\wei(v)= \frac{\xcor_2(v)-\xcor_2(\bar v)}{2}. 
\]
By (\ref{eq:x1x2}), we can see that (\ref{eq:weight01}) is equivalent to 
\[
0 \le \wei(u)+\wei(v) \le 2,  
\]
which shows that $\wei$, $\lb =0$, and $\ub = 2$ define $G$. 
Furthermore, for any $v\in V$, 
\begin{align*}
    v \in M &\iff v \bar v \in E' \\
            & \iff \xcor_1(v) \le \xcor_1(\bar v) \mbox{ and } \xcor_2(v) \ge \xcor_2(\bar v) \\
            & \iff 0 \le \wei(v) \le 1, 
\end{align*}
which shows that 
$M = \{v \in V \mid 0\le \wei(v) \le 1\}$.  
\end{proof}

To utilize Lemma~\ref{lem:key-lemma}, we need to find the set $M$ of mid-weight vertices;
that is, the vertices with weights in the range $[\lb/2, \ub/2]$.
The first observation is that $M$ has to be a clique
as the weight sum of any two vertices in $M$ is in the range $[\lb, \ub]$.
In the following, we show that an efficient maximum clique can be chosen as $M$.
To this end, we first prove that we only need to consider (inclusion-wise) maximal cliques.
\begin{lemma}
\label{lem:mid-weight_maximal-clique}
For a connected non-bipartite double-threshold graph $G = (V,E)$,
there exist $\wei\colon V \to \R$ and $\lb, \ub \in \R$ defining $G$
such that $\{v \in V \mid \lb/2 \le \wei(v) \le \ub/2\}$ is a maximal clique of $G$.
\end{lemma}
\begin{proof}
Let $G$ be a non-bipartite double-threshold graph $G = (V,E)$
defined by $\wei\colon V \to \R$ and $\lb, \ub \in \R$.
Let $M = \{v \in V \mid \lb/2 \le \wei(v) \le \ub/2\}$.
We choose $\wei$, $\lb$, and $\ub$ in such a way
that for any $\wei'\colon V \to \R$ and $\lb', \ub' \in \R$ defining $G$,
$M$ is not a proper subset of $\{v \in V \mid \lb'/2 \le \wei'(v) \le \ub'/2\}$.
Suppose to the contrary that $M$ is not a maximal clique of $G$.
Observe that if $M = \emptyset$, then $V$ can be partitioned into two independent sets
$\{v \in V \mid \wei(v) < \lb/2\}$ and $\{v \in V \mid \wei(v) > \ub/2\}$,
which is a contradiction to the non-bipartiteness of $G$.
Hence, $M$ is non-empty.

Let $G'_{M}$ be the auxiliary graph constructed from $G$ and $M$ as before.
By Lemma~\ref{lem:key-lemma}, $G'_{M}$ has a permutation diagram in which both 
${\prec_{1}} = \langle w_{1}, \dots, w_{2n} \rangle$ and
${\prec_{2}} = \langle w'_{1}, \dots, w'_{2n} \rangle$ are symmetric.
Let $\overline{M} = \{\bar{v} \mid v \in M\}$.
By the definition of $G'_{M}$, $M \cup \overline{M}$ induces a complete bipartite graph in $G'_{M}$.
By symmetry, we may assume that $M \prec_{1} \overline{M}$ and $\overline{M} \prec_{2} M$.
That is, in $\prec_{1}$ all vertices in $M$ appear before any vertex in $\overline{M}$ appears,
and in $\prec_{2}$ all vertices in $\overline{M}$ appear before any vertex in $M$ appears.
Note that these assumptions imply that for each edge $x \bar{y} \in E(G'_{M})$,
$x \prec_{1} \bar{y}$ and $\bar{y} \prec_{2} x$ hold
since $G'_{M}$ is connected by Lemma~\ref{lem:GMconn} (see \figref{fig:mid-weight} (Left)).

As $M$ is not a maximal clique in $G$,
there is a vertex $v \notin M$ such that $M \subseteq N_{G}(v)$.
If $\bar{v} \prec_{1} v$, then we have
\begin{align}
  M \prec_{1} \bar{v} \prec_{1} v \prec_{1} \overline{M} \quad \text{ and } \quad
  \bar{v} \prec_{2} \overline{M} \prec_{2} M \prec_{2} v \label{eq:maximal_the-1st-case}
\end{align}
since 
$v \bar{v} \notin E(G'_{M})$, $\overline{M} \subseteq N_{G'_{M}}(v)$, and $M \subseteq N_{G'_{M}}(\bar{v})$.
Similarly, if $v \prec_{1} \bar{v}$, then we have
\begin{align*}
  v \prec_1 M \prec_{1} \overline{M} \prec_{1} \bar{v} \quad \text{ and } \quad
  \overline{M} \prec_{2} v \prec_{2} \bar{v} \prec_{2} M,
\end{align*}
or equivalently,
\begin{align*}
  M \prec_{2}^{R} \bar{v} \prec_{2}^{R} v \prec_{2}^{R} \overline{M} \quad \text{ and } \quad
  \bar{v} \prec_{1}^{R} \overline{M} \prec_{1}^{R} M \prec_{1}^{R} v.
\end{align*}
Thus, by replacing $\prec_{1}$ with $\prec_{2}^{R}$ and $\prec_{2}$ with $\prec_{1}^{R}$ if necessary,
we may assume that \eqref{eq:maximal_the-1st-case} holds (see \figref{fig:mid-weight} (Left)).
We further assume that $v$ has the smallest position in $\prec_{1}$ under these conditions.
\begin{myclaim}
  \label{clm:center}
  $w_{n+1} = v$ (and thus $w_{n} = \bar{v}$).
\end{myclaim}
\begin{proof}
[Claim~\ref{clm:center}]
By the symmetry of $\langle w_{1}, \dots, w_{2n} \rangle$,
it suffices to show that there is no vertex $x \in V$ such that $\bar{v} \prec_{1} x \prec_{1} v$.
Suppose that such a vertex $x$ exists.
In $G'_{M}$, $x$ is not adjacent to $\bar{v}$. This implies that $xv \notin E$, and hence $x \notin M$.
On the other hand, in $G'_{M}$, $x$ is adjacent to all vertices in $\overline{M}$.
Thus, we have $M \subseteq N_{G}(x)$.
This contradicts that $v$ has the smallest position in $\prec_{1}$ under those conditions.
\end{proof}

Now we obtain $\prec'_{1}$ from $\prec_{1}$ by swapping $v$ and $\bar{v}$ (see \figref{fig:mid-weight} (Right)).
By Claim~\ref{clm:center}, this new ordering $\prec'_{1}$ gives (together with $\prec_{2}$)
the graph obtained from $G'_{M}$ by adding the edge $v\bar{v}$.
Observe that this new graph can be expressed as $G'_{M \cup \{v\}}$.
Since $\prec'_{1}$ and $\prec_{2}$ are symmetric, Lemma~\ref{lem:key-lemma} implies that 
there are
$\wei'\colon V \to \R$ and $\lb', \ub' \in \R$ defining $G$
such that $\{u \in V \mid \lb'/2 \le \wei'(u) \le \ub'/2\} = M \cup \{v\}$.
This contradicts the choice of $\wei$, $\lb$, and $\ub$.
\begin{figure}[t]
  \centering
  \includegraphics[scale=.55]{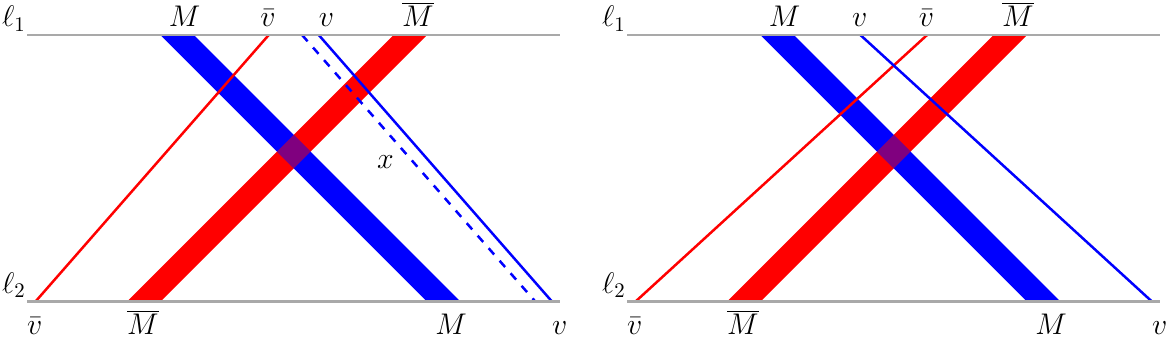}
  \caption{(Left) Relative positions of $v$, $\bar{v}$, $M$, and $\overline{M}$. 
    (Right) $\prec'_{1}$ is obtained from $\prec_{1}$ by swapping $v$ and $\bar{v}$.}
  \label{fig:mid-weight}
\end{figure}
\end{proof}

We show that every efficient maximum clique can be the set of mid-weight vertices,
given an appropriate choice of $\wei$, $\lb$, and $\ub$.

\begin{lemma}
\label{lem:efficient-max-clique}
Let $G$ be a non-bipartite double-threshold graph.
For every efficient maximum clique $K$ of $G$,
there exist $\wei\colon V \to \R$ and $\lb, \ub \in \R$ defining $G$ such that 
$K = \{v \in V \mid \lb/2 \le \wei(v) \le \ub/2\}$.
\end{lemma}
\begin{proof}
Let $K$ be an efficient maximum clique of $G$.
By Lemma~\ref{lem:DTGisPerm}, $G$ is a permutation graph, and thus cannot contain 
an induced odd cycle of length $5$ or more~\cite{Golumbic04}.
As $G$ is non-bipartite, $G$ contains $K_{3}$. This implies that $|K| \ge 3$.

By Lemma~\ref{lem:mid-weight_maximal-clique},
there exist $\wei\colon V \to \R$ and $\lb, \ub \in \R$ defining $G$
such that $M := \{v \in V \mid \lb/2 \le \wei(v) \le \ub/2\}$ is a maximal clique of $G$.
Assume that $\wei$, $\lb$, and $\ub$ are chosen so that the size of the symmetric difference
$|M \vartriangle K| = |M \setminus K| + |K \setminus M|$ is minimized.
Assume that $K \ne M$ since otherwise we are done.
This implies that $K \not\subseteq M$ and $K \not\supseteq M$ as both $K$ and $M$ are maximal cliques.
Observe that $G - M$ is bipartite.
This implies that $|K \setminus M| \in \{1,2\}$ and that $K \cap M \ne \emptyset$ as $|K| \ge 3$.
Since $K$ is a maximum clique, $|M \setminus K| \le |K \setminus M|$ holds.

Let $u \in K \setminus M$. By symmetry, we may assume that $\wei(u) < \lb/2$.
Note that no other vertex in $K$ has weight less than $\lb/2$ as $K$ is a clique.
Let $v \in M$ be a non-neighbor of $u$ that has the minimum weight among such vertices.
Such a vertex exists since $M$ is a maximal clique. Note that $v \in M \setminus K$.

We now observe that $v$ has the minimum weight in $M$.
If $w \in M$ is a non-neighbor of $u$, then $\wei(v) \le \wei(w)$ follows from the definition of $v$.
If $w \in M$ is a neighbor of $u$, then $\wei(v) < \wei(w)$ holds,
since otherwise $\wei(u) < \lb/2 \le \wei(w) \le \wei(v)$ and $uw, wv \in E$ imply that $uv \in E$ 
by Observation~\ref{obs:sandwich}.

We are going to show that $N(v) = N(u)$.

\begin{myclaim}
  \label{clm:neighborhood-in-M}
  $N(u) \cap \{x \mid \wei(x) < \lb/2\} = N(v) \cap \{x \mid \wei(x) < \lb/2\}  = \emptyset$.
\end{myclaim}
\begin{proof}
[Claim~\ref{clm:neighborhood-in-M}]
Since $\wei(u) < \lb/2$, $N(u) \cap \{x \mid \wei(x) < \lb/2\} = \emptyset$.
Suppose to the contrary that $v$ has a neighbor $x$ with $\wei(x) < \lb/2$.
The maximality of $M$ implies that $x$ has a non-neighbor $y \in M$.
Since $y \in M$, $\wei(v) \le \wei(y)$ holds.
However,
$\wei(x) < \lb/2 \le \wei(v) \le \wei(y)$ and $xv, vy \in E$ imply $xy \in E$ by Observation~\ref{obs:sandwich}.
\end{proof}

\begin{myclaim}
  \label{clm:neighborhood-in-M2}
  $N(u) \cap M = N(v) \cap M = M \setminus \{v\}$.
\end{myclaim}
\begin{proof}
[Claim~\ref{clm:neighborhood-in-M2}]
Since $M$ is a clique and $v \in M$, we have $N(v) \cap M = M \setminus \{v\}$.
Thus, the claim is equivalent to $M \setminus \{v\} \subseteq N(u)$.
This holds if $M \setminus K = \{v\}$.
Assume that $M \setminus K = \{v, v'\}$ for some $v' \ne v$.
To show the claim, it suffices to show that $uv' \in E$.

Since $|M \setminus K| \le |K \setminus M| \le 2$,
we have $K \setminus M = \{u,u'\}$ for some $u' \ne u$.
Since $\wei(u) < \lb/2$ and $uu' \in E$, we have $\wei(u') \ge \lb/2$. Moreover since $u' \notin M$, we have $\wei(u') > \ub/2$.
Let $w \in M \cap K$. If $\wei(w) > \wei(v')$,
then, by Observation~\ref{obs:sandwich}, we have $u'v, u'v' \in E$
since $\wei(v) \le \wei(v') < \wei(w) \le \wei(u')$
and $vw, v'w, wu' \in E$.
This implies that $M \subseteq N(u')$, which contradicts the maximality of $M$.
Hence, $\wei(w) \le \wei(v')$ holds.
This implies by Observation~\ref{obs:sandwich}
that $uv' \in E$ as $\wei(u) \le \wei(w) \le \wei(v')$ and $uw, wv' \in E$.
\end{proof}

\begin{myclaim}
  \label{clm:neighborhood-in-L}
  $N(u) \cap \{x \mid \wei(x) > \ub/2\} \supseteq N(v) \cap \{x \mid \wei(x) > \ub/2\}$.
\end{myclaim}
\begin{proof}
[Claim~\ref{clm:neighborhood-in-L}]
Let $w \in K \cap M$. For $z \in N(v)$ with $\wei(z) > \ub/2$, we have
\[
  \lb \le \wei(u) + \wei(w) \le \wei(u) + \ub/2 < \wei(u) + \wei(z) < \lb/2 + \wei(z) < \wei(v) + \wei(z) \le \ub,
\]
and thus $z \in N(u)$ holds.
\end{proof}

Claims~\ref{clm:neighborhood-in-M}, \ref{clm:neighborhood-in-M2}, and \ref{clm:neighborhood-in-L} imply that $N(v) \subseteq N(u)$.
To show that $N(v) = N(u)$, suppose to the contrary that 
$N(v)$ is a proper subset of $N(u)$.
We show that $K$ cannot be an efficient maximum clique in this case.
Let $K' = K \setminus \{u\} \cup \{v\}$. We first argue that $K'$ is a maximum clique.
To this end, it suffices to show that $K'$ is a clique as $|K'| = |K|$.
If $K \setminus M = \{u\}$, then $K' = M$ is a clique.
Assume that $K \setminus M = \{u, u'\}$ for some $u' \ne u$.
Since $\wei(u) < \lb/2$ and $u' \in K \setminus M$, we have $\wei(u') > \ub/2$ as before.
Let $w \in K \cap M$. Then, $vw, wu' \in E$.
Since $\wei(v) \le \wei(w) \le \ub/2 < \wei(u')$, we have $vu' \in E$ by Observation~\ref{obs:sandwich}.
Thus, $K'$ is a clique.
The assumption $N(v) \subsetneq N(u)$ implies that $\deg_{G}(v) < \deg_{G}(u)$, and thus,
\[
 \sum_{w \in K'} \deg_{G}(w) = \left(\sum_{w \in K} \deg_{G}(w)\right) - \deg_{G}(u) + \deg_{G}(v) < \sum_{w \in K} \deg_{G}(w).
\]
This contradicts that $K$ is efficient.
Therefore, we conclude that $N(v) = N(u)$.

Now, we define a weight function $\wei'\colon V \to \R$ by setting
$\wei'(u) = \wei(v)$, $\wei'(v) = \wei(u)$, and $\wei'(x) = \wei(x)$ for all $x \in V \setminus \{u,v\}$.
Then, $\wei'$, $\lb$, and $\ub$ define $G$ and 
$M' := \{w \in V \mid \lb/2 \le \wei'(w) \le \ub/2\} = M \cup \{u\} \setminus \{v\}$ as $N(u) = N(v)$.
This contradicts the choice of $\wei$ as $|M' \vartriangle K| < |M \vartriangle K|$.
\end{proof}

Next, we show that the symmetry required in Lemma~\ref{lem:key-lemma} follows for free
when $M$ is a clique.

\begin{lemma}
\label{lem:perm=sym}
Let $G = (V,E)$ be a connected non-bipartite graph and $M$ be a clique of $G$. 
Then, $G'_{M}$ is a permutation graph
if and only if
$G'_{M}$ can be represented by a permutation diagram in which both orderings $\prec_1$ and $\prec_2$ are symmetric.
\end{lemma}
\begin{proof}
The if part is trivial.
To prove the only-if part, we assume that $G'_{M}$ is a permutation graph.

First we observe that we only need to deal with the \emph{twin-free} case.
Assume that $N_{G'_{M}}(u) = N_{G'_{M}}(v)$ (or equivalently $N_{G'_{M}}(\bar{u}) = N_{G'_{M}}(\bar{v})$) for some $u, v \in V$, i.e., $u, v$ are twins in $G'_{M}$.
If $G'_{M} - \{v,\bar{v}\}$ has a permutation diagram in which both permutation orderings $\prec_{1}$ and $\prec_{2}$ are symmetric,
then we can obtain symmetric permutation orderings $\prec'_{1}$ and $\prec'_{2}$ of $G'_{M}$
by inserting $v$ right after $u$, and $\bar{v}$ right before $\bar{u}$ in both $\prec_{1}$ and $\prec_{2}$.
Thus, it suffices to show that $G'_{M} - \{v,\bar{v}\} = (G - v)'_{M \setminus \{v\}}$
has a permutation diagram in which both permutation orderings $\prec_{1}$ and $\prec_{2}$ are symmetric.

Observe that $G-v$ might be bipartite, but $(G - v)'_{M \setminus \{v\}}$ is still connected.
Hence, we can assume in the following 
that no pair of vertices in $G'_{M}$ have the same neighborhood
and that $G'_{M}$ is connected (but $G$ might be bipartite).
We also assume that $|V| \ge 2$ since otherwise the statement is trivially true.

Let $\prec_{1}$ and $\prec_{2}$ be the permutation orderings corresponding to a permutation diagram of $G'_{M}$.
By Lemma~\ref{lem:uniquerep}, the assumption of having no twins implies that
$\prec_{1}$, $\prec_{2}$, $\prec_{1}^{\mathrm{R}}$, and $\prec_{2}^{\mathrm{R}}$ are 
all the permutation orderings of $G'_{M}$.
Since $G'_{M}$ is connected,
we may assume that the first vertex in $\prec_{1}$ belongs to $V$,
the last  in $\prec_{1}$ belongs to $V' \setminus V$,
the first in $\prec_{2}$ belongs to $V' \setminus V$,
and the last vertex in $\prec_{2}$ belongs to $V$.
Let $\langle w_{1}, \dots, w_{2n} \rangle$ be the ordering defined by $\prec_{1}$.

Let $\varphi \colon V' \to V'$ be a map such that 
$\varphi(v) = \bar{v}$ and $\varphi(\bar{v}) = v$ for each $v \in V$.
This map $\varphi$ is an automorphism of $G'_{M}$.
Thus, $\langle \varphi(w_{1}), \dots, \varphi(w_{2n}) \rangle$ is also a permutation ordering of $G'_{M}$.
Let ${\prec'} = \langle \varphi(w_{1}), \dots, \varphi(w_{2n}) \rangle$ denote this ordering.
Then, 
\[
 {\prec'} \in \{\prec_{1}, \prec_{2}, \prec_{1}^{\mathrm{R}}, \prec_{2}^{\mathrm{R}}\}.
\]
We claim that ${\prec'} = {\prec_{1}^{\mathrm{R}}}$.
First, observe that ${\prec'} \notin \{\prec_{1},\prec_{2}^{\mathrm{R}}\}$ as
the first vertex of $\prec'$ belongs to $V' \setminus V$
but the first vertices of $\prec_{1}$ and $\prec_{2}^{\mathrm{R}}$ belong to $V$.

Suppose to the contrary that ${\prec'} = {\prec_{2}}$.
Then, for each $w \in V'$, the positions of $w$ in $\prec_{1}$ and $\varphi(w)$ in ${\prec_{2}}$ ($= {\prec'}$) are the same.
Thus, $w_{i} \prec_{1} \varphi(w_{i})$ implies $\varphi(w_{i}) \prec_{2} \varphi(\varphi(w_{i})) = w_{i}$.
Hence, we have $v \bar{v} \in E(G'_{M})$ for all $v \in V$, and thus $M = V$.
As $M$ is a clique, $M = V$ implies that $G$ is a complete graph $K_{|V|}$ and that $G'_{M}$ is a complete bipartite graph $K_{|V|, |V|}$.
This contradicts the assumption that $G'_{M}$ has no twins as $|V| \ge 2$.
Therefore, we conclude that ${\prec'} = {\prec_{1}^{\mathrm{R}}}$,
and in particular that $\varphi(w_{i}) = w_{2n-i+1}$ for each $i$.
This means that $w_{i} = v$ implies $w_{2n-i+1} = \bar{v}$ for all $v \in V$ and $i \in \{1,\dots,2n\}$.
Hence, $\prec_{1}$ is symmetric.
\end{proof}

Now we can prove Theorem~\ref{thm:non-bipartite} restated below.

\newcharacterization*
\begin{proof}
To show that \ref{itm:chara_dtg}$\implies$\ref{itm:every_emc},
assume that $G$ is a non-bipartite double-threshold graph.
Let $M$ be an efficient maximum clique of $G$.
By Lemma~\ref{lem:efficient-max-clique},
there exist $\wei\colon V \to \R$ and $\lb, \ub \in \R$ defining $G$ such that $M = \{v \in V \mid \lb/2 \le \wei(v) \le \ub/2\}$.
Now by Lemma~\ref{lem:key-lemma},
$G'_{M}$ is a bipartite permutation graph.

\smallskip

The implication \ref{itm:every_emc}$\implies$\ref{itm:some_emc} is trivial.

\smallskip

We now show that \ref{itm:some_emc}$\implies$\ref{itm:chara_dtg}.
Assume that for an efficient maximum clique $M$ of a non-bipartite graph $G$, the graph $G'_{M}$ is a bipartite permutation graph.

Let $H$ be a non-bipartite component of $G$.
Then, $H$ contains an induced odd cycle of length $k \ge 3$.
This means that, if $H$ does not contain $M$, 
then $G'_{M}$ contains an induced cycle of length $2k \ge 6$.
However, this is a contradiction as a permutation graph cannot contain
an induced cycle of length at least $5$~\cite{Gallai67}.
Thus, $H$ contains $M$. Also, there is no other non-bipartite component in $G$ as it does not intersect $M$.
Since $H$ contains $M$, $H'_{M}$ is a component of $G'_{M}$.
By Lemma~\ref{lem:perm=sym},
$H'_{M}$ can be represented by a permutation diagram in which both $\prec_1$ and $\prec_2$ are symmetric,
and thus $H$ is a double-threshold graph by Lemma~\ref{lem:key-lemma}.

Let $B$ be a bipartite component of $G$ (if one exists).
Since $B$ does not intersect $M$, $G'_{M}$ contains two isomorphic copies of $B$ as components.
Since $G'_{M}$ is a permutation graph, $B$ is a permutation graph too.
By Lemma~\ref{lem:bipartite}, $B$ is a double-threshold graph.

Now we know that all components of $G$ are double-threshold graphs
and exactly one of them is non-bipartite.
By Lemma~\ref{lem:bi-and-nonbi-components}, $G$ is a double-threshold graph.
\end{proof}


\section{Linear-time recognition algorithm}
\label{sec:algorithm}

We now present a linear-time recognition algorithm for double-threshold graphs.
\begin{theorem}
\label{thm:linear-recognition}
There is an $O(m+n)$-time algorithm that
accepts a given graph $G = (V,E)$ if and only if the graph is a double-threshold graph,
where $n = |V|$ and $m = |E|$.
\end{theorem}
\begin{proof}
Given a graph $G$, we accept $G$ if and only if 
\begin{itemize}
  \item $G$ is a bipartite permutation graph, or
  \item $G$ is a non-bipartite permutation graph and $G'_{M}$ is a permutation graph, where $M$ is an efficient maximum clique of $G$.
\end{itemize}
By Lemma~\ref{lem:bipartite} and Theorem~\ref{thm:non-bipartite}, this algorithm is correct.
Thus, it suffices to present a linear-time implementation of this algorithm.

We first test whether $G$ is a permutation graph in $O(m+n)$ time~\cite{McConnellS99}.
If $G$ is not a permutation graph, we can reject it by Lemma~\ref{lem:DTGisPerm}.
Otherwise, we check in linear time whether $G$ is bipartite.
If so, we can accept $G$ by Lemma~\ref{lem:bipartite}.

In the remaining case, $G$ is a non-bipartite permutation graph.
Assume for now that we already have an efficient maximum clique $M$ of $G$.
Since $|V(G'_{M})| = 2n$ and $|E(G'_{M})| = 2m + |M|$,
we can construct $G'_{M}$ and test whether it is a permutation graph in $O(m+n)$ time.
Hence, by Theorem~\ref{thm:non-bipartite}, it suffices to show that $M$ can be found in $O(m+n)$ time.

To find an efficient maximum clique of $G$,
we set to each vertex $v \in V$ the weight $f(v) = n^{2} - \deg_{G}(v)$,
and then find a maximum-weight clique of $G$ with respect to $f$.
%
It is known that a transitive orientation of a permutation graph can be computed in $O(m+n)$ time~\cite{McConnellS99},
and then using the orientation, we can find a maximum-weight clique $M$ in $O(m+n)$ time~\cite[pp.~133--134]{Golumbic04}.
We show that $M$ is an efficient maximum clique of $G$.
Let $K$ be an efficient maximum clique of $G$.
Since $\sum_{v \in K}f(v) \le \sum_{v \in M} f(v)$, we have
\begin{equation}
  |K| \cdot n^{2} - \sum_{v \in K} \deg_{G}(v) \le |M| \cdot n^{2} - \sum_{v \in M} \deg_{G}(v). \label{eq:efficient-max-clique}
\end{equation}
Since $0 \le \sum_{v \in S} \deg_{G}(v) < n^{2}$ for any $S \subseteq V$,
it holds that $  |K| \cdot n^{2} - n^{2} < |M| \cdot n^{2}$.
This implies that $|K| = |M|$ as $|K| \ge |M|$.
It follows from \eqref{eq:efficient-max-clique} that 
$\sum_{v \in K} \deg_{G}(v) \ge \sum_{v \in M} \deg_{G}(v)$.
Therefore, $M$ is an efficient maximum clique.
\end{proof}


\section{Conclusion}
\label{sec:conclusion}

We have presented a new characterization of double-threshold graphs
and a linear-time recognition algorithm for them based on the characterization.
For a better understanding of this graph class,
it would be good to have the list of minimal forbidden induced subgraphs.
We believe that our characterization will be useful for this direction as well.

\begin{acknowledgements}
The authors are grateful to Robert E. Jamison and Alan P. Sprague for sharing the man\-u\-script of their papers~\cite{JamisonS21,JamisonS20}.
The authors would also like to thank Martin Milani\v{c}, Gregory J. Puleo, and Vaidy Sivaraman for useful information about related papers.
The authors thank the anonymous reviewers for their constructive comments that considerably improved the presentation.
\end{acknowledgements}


\bibliographystyle{plainurl}
\bibliography{dtg}

\end{document}